\documentclass[11pt,centertags,reqno,twoside]{amsart}
\usepackage{amsmath,latexsym, graphicx}
\usepackage[psamsfonts]{amssymb}
\usepackage{mathptm}
\usepackage{color}
\usepackage{amsfonts}
\usepackage{amssymb}
\usepackage{amsthm}
\usepackage{bm}
\usepackage[english]{babel}
\usepackage[bookmarks]{hyperref}
\numberwithin{equation}{section}

\renewcommand{\epsilon}{\varepsilon}

\newcommand{\be}{\begin{equation}}
\newcommand{\ee}{\end{equation}}

\newcommand{\C}{\mathbb{C}}

\newcommand{\R}{\mathbb{R}}

\newcommand{\T}{\mathbb{T}}

\newcommand{\Z}{\mathbb{Z}}

\newcommand{\cC}{{\mathcal C}}

\newcommand{\cE}{{\mathcal E}}

\newcommand{\cH}{{\mathcal H}}

\renewcommand{\Re}{{\ensuremath{\mathrm{Re}}}}

\renewcommand{\det}{\mathop{\mathrm{det}}}

{\bf}{\it}
\newtheorem{theorem}{Theorem}[section]
\newtheorem{lemma}[theorem]{Lemma}

\newtheorem{remark}[theorem]{Remark}

\newcommand{\black}{\color{black}}

\date{}

\begin{document}
\large
\title[Two-fermion lattice Hamiltonian]
{Two-fermion lattice Hamiltonian with first and second
nearest-neighboring-site interactions}

\author[S.N.Lakaev, A.K.Motovilov, and
S.Kh.Abdukhakimov]{Saidakhmat N.\,Lakaev, Alexander K.\,Motovilov,
Saidakbar Kh.\,Abdukhakimov}

 \address{Saidakhmat N.\,Lakaev, Samarkand State University, Samarkand, 140104 Uzbekistan, and
Samarkand Branch of the Romanovskii Institute of Mathematics,
Academy of Sciences of the Republic of Uzbekistan, Samarkand, 140104
Uzbekistan} \email{slakaev@mail.ru}

\address{Alexander K. Motovilov, Bogoliubov Laboratory of
Theoretical Physics, JINR, Joliot-Cu\-rie 6, 141980 Dubna, Russia,
and Dubna State University, Universitetskaya 19, 141980 Dubna,
Russia} \email{motovilv@theor.jinr.ru}

\address{Saidakbar Kh.\,Abdukhakimov, Samarkand State University, Samarkand, 140104 Uzbekistan, and
Samarkand Branch of the Romanovskii Institute of Mathematics,
Academy of Sciences of the Republic of Uzbekistan, Samarkand, 140104
Uzbekistan} \email{abduxakimov93@mail.ru}

\begin{abstract}
We study the Schr\"odinger operators ${H}_{\lambda \mu}(K)$, with
$K\in\T^2$ the fixed quasimomentum of the particles pair, associated
with a system of two identical fermions on the two-dimen\-sional
lattice $\mathbb{Z}^2$ with first and second
nearest-neighboring-site interactions of magnitudes $\lambda\in\R$
and $\mu\in\R$, respectively. We establish a partition of the
$(\lambda,\mu)-$plane so that in each its connected component, the
Schr\"{o}dinger operator ${H}_{\lambda\mu}(0)$ has a definite
(fixed) number of eigenvalues, which are situated below the bottom
of the essential spectrum and above its top. Moreover, we establish
a sharp lower bound for the number of isolated eigenvalues of
${H}_{\lambda\mu}(K)$ in each connected component.
\vspace*{-0.5cm}

\end{abstract}

\maketitle

\black
\section{\textbf{Introduction}}\label{sec:intro}

Lattice models play an important role in various branches of
physics. Among such models are the lattice few-body Hamiltonians
\cite{Mattis:1986} that may be viewed as a minimalist version of the
corresponding Bose- or Fermi-Hubbard model involving a fixed finite
number of particles of a certain type. Surely, the few-body lattice
Hamiltonians are of a great theoretic interest already in their own
right
\cite{ALzM:2004,ALKh:2012,BPL:2017,KhLA:2021,Lakaev:1993,%
LAbdukhakimov:2020,LDKh:2016,LLakaev:2017,LOzdemir:2016}. In
addition, these discrete Hamiltonians may be viewed as a natural
approximation for their continuous counterparts
\cite{FMerkuriev:1993} allowing to study few-body phenomena in the
context of the theory of bounded operators. A still intriguing
phenomenon is the celebrated Efimov effect \cite{Efimov:1970} which
is proven to take place not only in the continuous case but also in
the lattice three-body problems \cite{ALzM:2004,
ALKh:2012,DzMSh:2011,Lakaev:1993}. Furthermore, the discrete
Schr\"odinger operators represent the simplest and natural model for
description of few-body systems formed by particles traveling
through periodic structures, say, for ulracold atoms injected into
optical crystals created by the interference of counter-propagating
laser beams \cite{Bloch:2005, W-Z:2006}. The study of ultracold
few-atom systems in optical lattices became very popular in the last
years since these systems possess highly controllable parameters
such as lattice geometry and dimensionality, particle masses,
two-body potentials, temperature etc.  (see e.g., \cite{Bloch:2005,
J-Z:1998,JZoller:2005,LSAhufinger:2012} and references therein).
Unlike the traditional condensed matter systems, where stable
composite objects are usually formed by attractive forces, the
controllability of the ultracold atomic systems in an optical
lattice gives an opportunity to experimentally observe a stable
{repulsive} bound pair of ultracold atoms, see e.g.,
\cite{OOBongs:2006, W-Z:2006}. Already one-particle one-dimensional
lattice Hamiltonians are of interest in applications. For example,
in \cite{MSBelyaev:2001}, effectively a one-dimensional one-particle
lattice Hamiltonian has been employed to exhibit explicitly how an
arrangement of molecules of a certain class in lattice structures
may enhance the nuclear fusion probability.

Unlike in the continuous case, the lattice few-body system does not
admit separation of the center-of-mass motion. However, the discrete
translation invariance allows one to use the Floquet-Bloch
decomposition (see, e.g., \cite[Sec. 4]{ALMM:2006}). In particular,
the total $n$-particle lattice Hamiltonian $\mathrm{H}$ in the
(quasi)momentum representation may be written as the von Neumann
direct integral
\begin{equation}
\label{HK} \mathrm{H}\simeq\int\limits_{K\in \T^d} ^\oplus  H(K)\,d
K,
\end{equation}
where $\T^d$ is the $d$-dimensional torus and $H(K)$, the fiber
Hamiltonian acting in the respective functional Hilbert space on
$\T^{(n-1)d}$. Recall that the fiber Hamiltonians $H(K)$
nontrivially depend on the quasimomentum $K\in\T^d$ (see e.g.,
\cite{ALMM:2006, FIC:2002,Mattis:1986,Mogilner:1991}).

It is well known that the Efimov effect \cite{Efimov:1970} that  we
already mentioned before, was originally attributed to the
three-body systems moving in the three-dimensional continuous space
$\mathbb{R}^3$. \black The essence of the effect is as follows. A
system of three particles in $\mathbb{R}^3$ with pairwise attractive
short-range potentials has an infinite number of binding energies
exponentially converging to zero if the two-particle subsystems do
not have a negative spectrum and at least two of them are resonant
in the sense that any arbitrarily small negative perturbation of the
two-body interaction produces a negative spectrum. A rigorous
mathematical proof of the Efimov effect has been given in
\cite{Ovchinnikov:1979,Sobolev:1993,Tamura:1991,Yafaev:1974}. In
\cite{ALzM:2004,ALKh:2012,Lakaev:1993}, the existence of the Efimov
effect has also been proven for three-body systems on the
three-dimensional lattice $\mathbb{Z}^3$. Later on, the existence of
Efimov-type phenomena has been predicted by physicists for a
five-boson system moving on a line $\mathbb{R}$ \cite{Nishida:2010},
for a four-boson system on the plane $\mathbb{R}^2$
\cite{Nishida:2017}, and for a system of three spinless fermions
moving on the plane $\mathbb{R}^2$ \cite{Nishida:2013}. In the
latter case, a mathematical proof is available
\cite{Gridnev:2014,Tamura:2019}, and the phenomenon acquired the
name of a super Efimov effect, because of the double exponential
convergence of the binding energies to the three-body threshold. One
may guess that a similar phenomenon should take place in the system
of three spinless fermions on the two-dimensional lattice
$\mathbb{Z}^2$, at least for some values of the center-of-mass
quasimomentum. Surely, in order to prove or disprove this, one needs
first to study properties of the system of two spinless fermions on
the lattice $\mathbb{Z}^2$. In the present work, we are making the
first step on this path and study the way how new eigenvalues emerge
from the lower and/or upper thresholds of the essential (continuous)
spectrum of the fiber Hamiltonians $H(K)$ involved.

In order to obtain a more detail information, we consider an
interaction between particles that contains two terms, the one which
is only non-trivial if the particles are located in the nearest
neighboring sites of the lattice, and another one, only non-trivial
if the particles are positioned in the next to nearest neighboring
sites (see Sec.  \ref{subsection_position}, definition
\eqref{def:potentials}). These terms include real factors (coupling
constants) $\lambda$ and $\mu$, respectively, and, in the
(quasi)momentum representation, the combined interaction potential
is denoted by $V_{\lambda\mu}$. The presence in $V_{\lambda\mu}$ of
the two terms independent of each other and each depending on the
corresponding parameter $\lambda$ and $\mu$ allows the fiber
Hamiltonian $H(K)$ to have eigenvalues simultaneously below and
above the essential spectrum.

Thus, as the entries $H(K)$ in \eqref{HK}, in this work we study the
family of the fiber Hamiltonians
\begin{equation}
H_{\lambda\mu}(K):=H_0(K) + V_{\lambda\mu},\qquad
K=(K_1,K_2)\in\T^2, \label{Hlm}
\end{equation}
where $H_0(K)$ is the fiber kinetic-energy operator,
$$
\bigl(H_0(K)f\bigr)(p)=\cE_K(p)f(p), \quad p=(p_1,p_2)\in\T^2, \quad
 f\in L^{2,o}(\T^2),\black
$$
with
\begin{equation}\label{def:dispersion}
\cE_K(p):= 2 \sum_{i=1}^2\Big(1-\cos\tfrac{K_i}2\,\cos p_i\Big)
\end{equation}
The potential  $V_{\lambda\mu}$ is an integral operator on
 $L^{2,o}(\T^2)$  with a smooth kernel function explicitly given by
formula \eqref{moment_poten} below. The formula \eqref{moment_poten}
implies that for any non-zero $\lambda,\mu$ the operator
$V_{\lambda\mu}$ is rank 6. Notice that $V_{\lambda\mu}$ does not
depend on $K$ at all. Surely, the operators $H_0(K)$ and
$V_{\lambda\mu}$, $\lambda,\mu\in\R$, are both bounded and
self-adjoint. Since $V_{\lambda\mu}$ is finite rank, the essential
spectrum  of  $H_{\lambda\mu}(K)$  coincides with that of $H_0(K)$
(see Sec. \ref{subsec:ess_spec}), i.e. it  coincides with the
segment $[\cE_{\min}(K),\allowbreak\cE_{\max}(K)],$ where
$$
\cE_{\min}(K):= 2\sum\limits_{i=1}^2\Big(1-\cos \tfrac{K_{i}}2\Big),
\qquad \cE_{\max}(K):= 2\sum\limits_{i=1}^2\Big(1+\cos
\tfrac{K_{i}}2\Big).
$$

To the best of our knowledge, the Hamiltonian \eqref{Hlm} represents
a new exactly solvable model. Within this model, we will first find
both the exact number and location of eigenvalues of the edge
operator $H_{\lambda\mu}(0)$. Then, for any pair of the interaction
parameters $\lambda,\mu\in\R$, we will establish sharp lower bounds
on the numbers of isolated eigenvalues of ${H}_{\lambda\mu}(K),K\neq
0$ lying below and above the essential spectrum Theorem
\ref{teo:disc_quasiK} and \ref{teo:disc_quasi0}.

Main goal of the article is to understand the mechanism of emergence
of eigenvalues of $H_{\lambda\mu}(K)$ from the essential spectrum as
$\lambda$ and $\mu$ vary as well as to clarify the inverse process,
the absorption of eigenvalues by the essential spectrum. To achieve
this goal, we use as a technical tool, the  Fredholm determinants
\cite{Albeverio:1988,Lakaev:1989}. Namely, we consider the Fredholm
determinant $D_{\lambda\mu}(K,z)$ associated with the
Lippmann-Schwinger operator generated by the unperturbed Hamiltonian
$H_0(K)$ and perturbation $V_{\lambda\mu}$.  It is well known
\cite{Albeverio:1988} that for any fixed $K\in\T^2$, there is a
one-to-one mapping between the set of eigenvalues of the perturbed
operator $H_{\lambda\mu}(K)=H_0(K)+V_{\lambda\mu}$ and the set of
zeros of the associated determinant $D_{\lambda\mu}(K,z)$.

\tolerance 1000 We start with a study of the properties of the the
Fredholm determinant $\Delta_{\lambda\mu}(z):=D_{\lambda\mu}(0,z)$
in the edge case $K=0$. Assuming that $C^{-}(\lambda,\mu)$ (resp.
$C^{+} (\lambda,\mu)$) is the main (constant) term of the
asymptotics of the function $\Delta_{\lambda\mu}(z)$ as $z$
converges to the lower (resp. upper) threshold of the essential
spectrum, we show that an additional root of
$\Delta_{\lambda\mu}(z)$ emerges below (resp. above) the essential
spectrum of $H_{\lambda\mu}(0)$ if and only if
$C^{-}(\lambda,\mu)=0$ (resp. $C^{+} (\lambda,\mu)=0$) (see Lemmas
\ref{eigen-zeros} and \ref{lemm:asympt}). Therefore, the number of
eigenvalues  of 1$H_{\lambda\mu}(0)$  changes if and only if the
point $(\lambda,\mu)$ on the parameter plane $\R^2$ crosses one of
the {curves  $C^{-}(\lambda,\mu)=0$ or $C^{+}(\lambda,\mu)=0$} (see
also Lemmas \ref{simple1}-\ref{simple3}). Moreover, after each such
a single  crossing, the number of eigenvalues of
$H_{\lambda\mu}(0)$ changes exactly by one. Surely,   this crossing
event is interpreted as a moment when the essential spectrum of
$H_{\lambda\mu}(0)$ either gives birth to or absorbs a bound state
of $H_{\lambda\mu}(0)$(see Theorem \ref{teo:disc_quasi0}).
Furthermore, the curves $C^{-}(\lambda,\mu)=0$ and $C^{+}
(\lambda,\mu)=0$ divide the $(\lambda,\mu)$ parameter plane into
several simply connected domains on each of which the number of
eigenvalues of the operator $H_{\lambda\mu}(0)$ remains constant
(see Theorem \ref{teo:constant}).

We notice that in \cite{HzMK:2020,LKhKh:2021} similar results were
obtained for a lattice two-boson system. In that case, the
description of the partition of the parameter plane $\R^2$ into the
connected components is quite elementary. In the present, fermionic,
case the description of partition of $\R^2$  into the connected
components is much more complicated and requires a special
technique. Surprisingly, the maximum number of isolated eigenvalues
is achieved only in four connected components where both $\lambda$
and $\mu$ run through infinite intervals.

The paper is organized as follows. In Sec. \ref{section_two_fermion}
we introduce the two-particle Hamiltonian in the position and
quasimomentum representations. Sec. \ref{sec:main_results} contains
statements of our main results. In Sec. \ref{sec:auxiliary} we
present some auxiliary facts that are needed in the proofs of the
main results. These proofs themselves are presented in Sec.
\ref{sec:proofs}. For convenience of the reader, in Appendix
\ref{sec:append_A} we give the proof of Lemma~\ref{asympt(abcdef)}.

\section{\textbf{Hamiltonian of a lattice two-fermion system}}
\label{section_two_fermion}

\subsection{The two-fermion Hamiltonian in the position--space representation}
\label{subsection_position}

Let $\Z^{2}$ be the two-dimensional lattice and $\ell
^{2,a}(\Z^{2}\times\Z^{2})\subset\ell ^{2}(\Z^{2}\times\Z^{2})$, the
Hilbert space of square--summable antisymmetric functions on
$\Z^{2}\times\Z^{2}$.

In the position-space representation, the Hamiltonian
$\widehat{\mathbb{H}}_{\lambda\mu}$ associated with a system of  two
fermions with a first and second nearest-neighboring-site
interaction potential $\widehat{\mathbb{V}}_{\lambda\mu}$ is an
operator on $\ell ^{2,a}(\Z^{2}\times\Z^{2})$ of the following form:
\begin{equation}\label{two_total}
\widehat{\mathbb{H}}_{\lambda\mu}=\widehat{\mathbb{H}}_{0}+
\widehat{\mathbb{V}}_{\lambda\mu},\,\,\lambda,\mu\in\mathbb{R}.
\end{equation}
Here, $\widehat{\mathbb{H}}_{0}$ is the kinetic energy operator of
the system, defined on  $\ell ^{2,a}(\Z^{2}\times\Z^{2})$  as
\begin{equation}
[\widehat{\mathbb{H}}_{0}\hat{f}](x_1,x_2)= \sum_{s_1\in\Z^2}
\hat{\varepsilon}(x_1-s_1)\hat{f}(s_1,x_2)+\sum_{s_2\in\Z^2}
\hat{\varepsilon}(x_2-s_2)\hat{f}(x_1,s_2),\,\,\hat{f}\in\ell
^{2,a}(\Z^{2}\times\Z^{2}), \label{two_free}
\end{equation}
where
\begin{equation}\label{def:epsilon}
  {\hat{\varepsilon}(s)=}
\left\lbrace\begin{array}{ccc}
2,\quad \quad    |s|=0,\\
-\frac{1}{2},\,\, \quad  |s|=1,\\
0,\quad \quad   |s|>1,
\end{array}\right.
\end{equation}
with $|s|=|s_{1}|+|s_{2}|$ for $s=(s_{1},s_{2})\in \mathbb{Z}^2$.
The first and second nearest-neighboring-site interaction potential
$\widehat{\mathbb{V}}_{\lambda\mu}$ is the operator of
multiplication by a function $\hat{v}_{\lambda\mu}$,
\begin{equation}\label{interaction}
[\widehat{\mathbb{V}}_{\lambda\mu}\hat{f}](x_1,x_2)=
\hat{v}_{\lambda\mu}(x_1-x_2)\hat{f}(x_1,x_2),\,\,\hat{f}\in\ell
^{2,a}(\Z^{2}\times\Z^{2}),
\end{equation}%
where
\begin{equation}\label{def:potentials}
\hat{v}_{\lambda\mu}(s)= \left\lbrace\begin{array}{ccc}
 \frac {\lambda}{2}, \quad \quad |s|=1,\\
\frac {\mu}{2},\quad \quad |s|=2,\\
0, \quad \quad |s|>2.
\end{array}\right.
\end{equation}
Obviously, all the three operators $\widehat{\mathbb{H}}_{0}$,
$\widehat{\mathbb{V}}_{\lambda\mu}$, and
$\widehat{\mathbb{H}}_{\lambda\mu}$ (for $\lambda,\mu\in\R$) are
bounded and self-adjoint.

\subsection{The two-fermion Hamiltonian in the quasimomentum representation}

Let $\mathbb{T}^{2}$ be the two-dimensional torus,
$\mathbb{T}^{2}=(\mathbb{R}/2\pi\mathbb{Z)}^{2} \equiv \lbrack -\pi
,\pi )^{2}$. The torus $\mathbb{T}^{2}$ represents the Pontryagin
dual group of $\mathbb{Z}^{2}$, equipped with the  Haar measure
$\mathrm{d}p$. Let $L^{2,a}(\T^2\times\T^2)$ be the Hilbert space of
square-integrable antisymmetric functions on $\T^2\times\T^2.$

The quasimomentum-space version of the Hamiltonian \eqref{two_total}
reads as
$$
\mathbb{H}_{\lambda\mu}:=(\mathcal{F}\otimes \mathcal{F}) \hat
{\mathbb{H}}_{\lambda\mu}(\mathcal{F}\otimes \mathcal{F})^*,
$$
where $\mathcal{F}\otimes \mathcal{F}$ denotes the Fourier
transform. The operator $\mathbb{H}_{\lambda\mu}$ acts on
$L^{2,a}(\T^2\times \T^2)$ and has the form
$\mathbb{H}_{\lambda\mu}=\mathbb{H}_0 + \mathbb{V}_{\lambda\mu}$,
where $\mathbb{H}_0=(\mathcal{F}\otimes \mathcal{F}) \hat
{\mathbb{H}}_0 (\mathcal{F}\otimes \mathcal{F})^*$ is the
multiplication operator:
$$
[\mathbb{H}_0 f](p,q) = [\epsilon(p) + \epsilon(q)]f(p,q),
$$
with
$$
\epsilon(p) := \sum\limits_{i=1}^2 \big(1-\cos p_i),\quad
p=(p_1,p_2)\in \T^2,
$$
the \emph{dispersion relation} of a single fermion.
The interaction  $\mathbb{V}_{\lambda\mu}=(\mathcal{F}\otimes
\mathcal{F})\hat{\mathbb{V}}_{\lambda\mu} (\mathcal{F}\otimes
\mathcal{F})^*$ is the integral operator
$$
[\mathbb{V}_{\lambda\mu} f](p,q) = \frac{1}{(2\pi)^2}\int_{\T^2}
v_{\lambda\mu}(p-u) f(u,p+q-u)\mathrm{d} u
$$
with the kernel function
$$
v_{\lambda\mu}(p)=\lambda\sum_{i=1}^2\cos p_i+\mu\sum_{i=1}^2\cos
2p_i+2\mu\sum\limits_{i=1}^2\sum\limits_{\,i\ne j=1}^2\cos p_i\cos
p_j,\quad p=(p_1,p_2)\in \T^2.
$$

\subsection{The Floquet-Bloch decomposition of
$\mathbb{H}_{\lambda\mu}$ and discrete Schr\"odinger
operators}\label{subsec:von_neuman}

Since $\hat{\mathbb{H}}_{\lambda\mu}$ commutes with the
representation of the discrete group $\Z^2$ by shift operators on
the lattice, the space $L^{2,a}(\T^2\times\T^2)$ and
$\mathbb{H}_{\lambda\mu}$ can be decomposed into the von Neumann
direct integral as (see, e.g., \cite{ALMM:2006})
\begin{equation}\label{hilbertfiber}
L^{2,a}(\T^2\times \T^2)\simeq \int_{K\in \T^2}^\oplus \black
L^{2,o}(\T^2)\,\mathrm{d}K
\end{equation}
and
\begin{equation}\label{fiber}
\mathbb{H}_{\lambda\mu} \simeq \int_{K\in \T^2}^\oplus \black
H_{\lambda\mu}(K)\,\mathrm{d} K,
\end{equation}
where $L^{2,o}(\T^2)$ is the Hilbert space of odd functions on
$\T^2$. The fiber operator $H_{\lambda\mu}(K),$ $K\in\T^2$, in
\eqref{fiber} acting on $L^{2,o}(\T^2)$  is of the form
\begin{equation}\label{momentum}
H_{\lambda\mu}(K) := H_0(K) + V_{\lambda\mu},
\end{equation}
where the (unperturbed) operator $H_0(K)$ is the multiplication
operator  by the function \eqref{def:dispersion} and the
perturbation operator  $V_{\lambda\mu}$ is given by
\begin{align}\label{moment_poten}
[V_{\lambda\mu}f](s)=&\frac{\lambda}{(2\pi)^2}\sum\limits_{i=1}^2\sin
s_i\int\limits_{\mathbb{T}^2} \sin t_i
f(t)\mathrm{d}t+\frac{\mu}{(2\pi)^2}\sum\limits_{i=1}^2\sin 2s_i
\int\limits_{\mathbb{T}^2}\sin 2t_if(t)\mathrm{d}t\\
&+\frac{\mu}{2\pi^2}\sum\limits_{i=1}^2\sum\limits_{\,i\ne
j=1}^2\sin s_i \cos s_j\int\limits_{\mathbb{T}^2}\sin t_i \cos t_j
f(t)\mathrm{d}t.\nonumber
\end{align}
Obviously, both the operators $H_0(K)$ and $V_{\lambda\mu}$ are
bounded and self-adjoint. In the literature, the parameter
$K\in\T^2$ is called the  \emph{two-particle quasimomentum} and the
entry $H_{\lambda\mu}(K)$ is called the \emph{discrete Schr\"odinger
operator} associated to the two-particle Hamiltonian
$\hat{\mathbb{H}}_{\lambda\mu}.$

\subsection{The essential spectrum of discrete Schr\"odinger operators}
\label{subsec:ess_spec}

Depending on $\lambda, \mu \in \R$, the rank of $V_{\lambda\mu}$
varies but never exceeds six. Hence, by Weyl's theorem,  for any
$K\in\T^2$ the essential spectrum
$\sigma_{\mathrm{ess}}(H_{\lambda\mu}(K))$  of $H_{\lambda\mu}(K)$
coincides with the spectrum of $H_0(K),$ i.e.,
\begin{equation}\label{eq:essential_spectrum}
\sigma_{\mathrm{ess}}(H_{\lambda\mu}(K))=\sigma(H_0(K)) =
[\cE_{\min}(K),\cE_{\max}(K)],
\end{equation}
with
\begin{align*}
\cE_{\min}(K):= & \min_{p\in  \T ^2}\,\cE_K(p) =
2\sum\limits_{i=1}^2\Big(1-\cos \tfrac{K_{i}}2\Big)\geq \cE_{\min}(0)=0,\\
\cE_{\max}(K):= & \max_{p\in  \T ^2}\,\cE_K(p)=
2\sum\limits_{i=1}^2\Big(1+\cos \tfrac{K_{i}}2\Big)\leq
\cE_{\max}(0)=8,
\end{align*}
where
\begin{equation}\label{dispersion}
\cE_K(p):= 2\sum_{i=1}^2\Big(1-\cos\tfrac{K_i}2\,\cos p_i\Big).
\end{equation}

\section{Main results}\label{sec:main_results}
Our first  main result is the following generalization of Theorems 1
and 2 in \cite{ALMM:2006}.

\begin{theorem}\label{teo:discr_Kne0}
Suppose that ${H}_{\lambda\mu}(0)$ has $n$ eigenvalues below (resp.
above) the essential spectrum for some $\lambda,\mu\in\R.$ Then for
each $K\in\T^2$ the operator ${H}_{\lambda\mu}(K)$ has at least $n$
eigenvalues below (resp. above) its essential  spectrum.
\end{theorem}
Denote by $\mu^{\pm}_{0}$ and $\mu^{\pm}_{1}$ the following numbers:
\begin{align}\label{root01}
&\mu^{\pm}_{0}=\frac{88-30\pi\pm
\sqrt{1044\pi^2-6720\pi+10816}}{240\pi-24\pi^2-512}\pi,
\end{align}
and
\begin{align}
\label{root02} &\mu^{\pm}_{1}=\frac{128-16\pi-9\pi^2\pm\sqrt{
225\pi^4-1440\pi^3+3904\pi^2-10240\pi+16384}}{120\pi-12\pi^2-256}.
\end{align}

Note that the numerical values of $\mu^{\pm}_{0}$ and
$\mu^{\pm}_{1}$ are as follows:
$$
\mu_0^-=-5.6172...,\quad \mu_0^+=-2.0623..., \quad
\mu_1^-=-5.7523..., \quad \mu_1^+=-2.9272...,
$$
and, hence, these numbers satisfy the relations
\begin{align}\label{roots}
&\mu^{-}_{1}<\mu_{0}^{-} <\mu^{+}_{1}<\mu^{+}_{0}<0.
\end{align}

By using the numbers  $\mu_0^+,\mu_0^-$  and $\mu_1^+,\mu_1^-$
defined, respectively, in \eqref{root01} and \eqref{root02} we
introduce the following two functions on $\mathbb{R}^2$:
\begin{align}
\label{abobel}
&C^{\pm}(\lambda,\mu)=\frac{30\pi-3\pi^2-64}{6\pi^2}[8(\mu\pm\mu_0^+)(\mu\pm\mu_0^-)
\mp\lambda(\mu\pm\mu_1^+)(\mu\pm\mu_1^-)].
\end{align}
Obviously, the curves on $(\lambda,\mu)$-plane defined by the
equations $C^{\pm}(\lambda,\mu)= 0$ coincide with the graphs of the
respective functions
\begin{equation}\label{curves}
\lambda(\mu)= \pm\frac{8(\mu\pm\mu_{0}^+)(\mu\pm\mu_{0}^-)}
{(\mu\pm\mu_1^+)(\mu\pm\mu_1^-)}.\\
\end{equation}
Any of the two functions \eqref{curves} is differentiable on its
domain. The graph of each of them consists of three separate smooth
curves with the respective asymptotes $\mu=\mp\mu_1^+$ and
$\mu=\mp\mu_1^-$. In each case these separate curves divide the
plane $\mathbb{R}^2$ into the four non-overlapping connected
components (see Figs. \ref{fig:1} and \ref{fig:dynamics2})
\begin{align*}
\cC_{0}^{-}&=\{(\lambda,\mu)\in\R^2:
\lambda>-\frac{8(\mu-\mu_{0}^+)(\mu-\mu_{0}^-)}
{(\mu-\mu_1^+)(\mu-\mu_1^-)},\,\,\mu>\mu_1^+\},\\
\cC_{1}^{-}&=\{(\lambda,\mu)\in\R^2:\,\,\lambda<-\frac{8(\mu-\mu_{0}^+)(\mu-\mu_{0}^-)}
{(\mu-\mu_1^+)(\mu-\mu_1^-)},\,\,
\mu>\mu_1^+\}\\
&\quad\quad\cup\{(\lambda,\mu)\in\R^2:\,\, \lambda\in\R,\,\,\mu=\mu_1^+\}\\
&\quad\quad\cup\{(\lambda,\mu)\in\R^2:\,\,\lambda>-\frac{8(\mu-\mu_{0}^+)(\mu-\mu_{0}^-)}
{(\mu-\mu_1^+)(\mu-\mu_1^-)},\,\,
\mu_1^-<\mu<\mu_1^+\},\\
\cC_{2}^{-}&=\{(\lambda,\mu)\in\R^2:\,\lambda<-\frac{8(\mu-\mu_{0}^+)(\mu-\mu_{0}^-)}
{(\mu-\mu_1^+)(\mu-\mu_1^-)}
,\,\,\,\,\mu_1^-<\mu<\mu_1^+\}\\
&\quad\quad\cup\{(\lambda,\mu)\in\R^2:\,\, \lambda\in\R,\,\,\mu=\mu_1^-\}\\
&\quad\quad\cup\{(\lambda,\mu)\in\R^2:\,\lambda>-\frac{8(\mu-\mu_{0}^+)(\mu-\mu_{0}^-)}
{(\mu-\mu_1^+)(\mu-\mu_1^-)}
,\,\,\,\,\mu< \mu_1^-\},\\
\cC_{3}^{-}&=\{(\lambda,\mu)\in\R^2:\,\lambda<-\frac{8(\mu-\mu_{0}^+)(\mu-\mu_{0}^-)}
{(\mu-\mu_1^+)(\mu-\mu_1^-)},\,\,\,\,\mu<\mu_1^-\}
\end{align*}
and
\begin{align*}
\cC_{0}^{+}&=\{(\lambda,\mu)\in\R^2:
\lambda<\frac{8(\mu+\mu_{0}^+)(\mu+\mu_{0}^-)}
{(\mu+\mu_1^+)(\mu+\mu_1^-)},\,\,\mu<-\mu_1^+\},\\
\cC_{1}^{+}&=\{(\lambda,\mu)\in\R^2:\,\,\lambda>\frac{8(\mu+\mu_{0}^+)(\mu+\mu_{0}^-)}
{(\mu+\mu_1^+)(\mu+\mu_1^-)},\,\,
\mu<-\mu_1^+\}\\
&\quad\quad\cup\{(\lambda,\mu)\in\R^2:\,\, \lambda\in\R,\,\,\mu=-\mu_1^+\}\\
&\quad\quad\cup\{(\lambda,\mu)\in\R^2:\,\,\lambda<\frac{8(\mu+\mu_{0}^+)(\mu+\mu_{0}^-)}
{(\mu+\mu_1^+)(\mu+\mu_1^-)},\,\,-\mu_1^+<\mu<-\mu_1^-\},\\
\cC_{2}^{+}&=\{(\lambda,\mu)\in\R^2:\,\,\lambda>\frac{8(\mu+\mu_{0}^+)(\mu+\mu_{0}^-)}
{(\mu+\mu_1^+)(\mu+\mu_1^-)},\,\,-\mu_1^+<\mu<-\mu_1^-\}\\
&\quad\quad\cup\{(\lambda,\mu)\in\R^2:\,\, \lambda\in\R,\,\,\mu=-\mu_1^-\}\\
&\quad\quad\cup\{(\lambda,\mu)\in\R^2:\,\lambda<\frac{8(\mu+\mu_{0}^+)(\mu+\mu_{0}^-)}
{(\mu+\mu_1^+)(\mu+\mu_1^-)}
,\,\,\,\,\mu> -\mu_1^- \},\\
\cC_{3}^{+}&=\{(\lambda,\mu)\in\R^2:\,\lambda>\frac{8(\mu+\mu_{0}^+)(\mu+\mu_{0}^-)}
{(\mu+\mu_1^+)(\mu+\mu_1^-)},\,\,\,\,\mu>-\mu_1^-\}.
\end{align*}

\begin{figure}[h!]
    \centering
\includegraphics[width=0.9\textwidth]{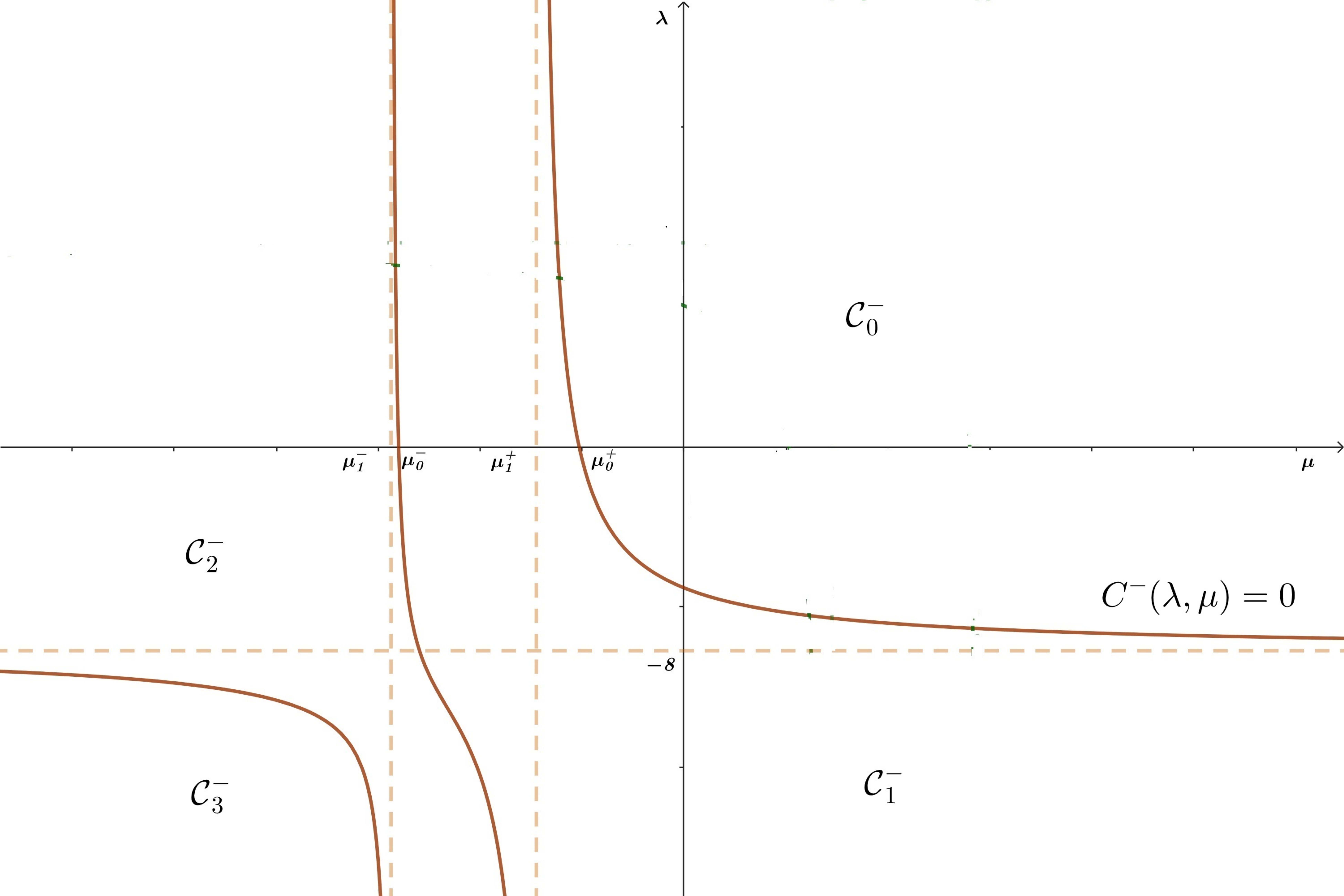}
\caption{Plot of the curves defined by equation
$C^-(\lambda,\mu)=0$}\label{fig:1}
\end{figure}
\begin{figure}[h!]
    \centering
\includegraphics[width=0.9\textwidth]{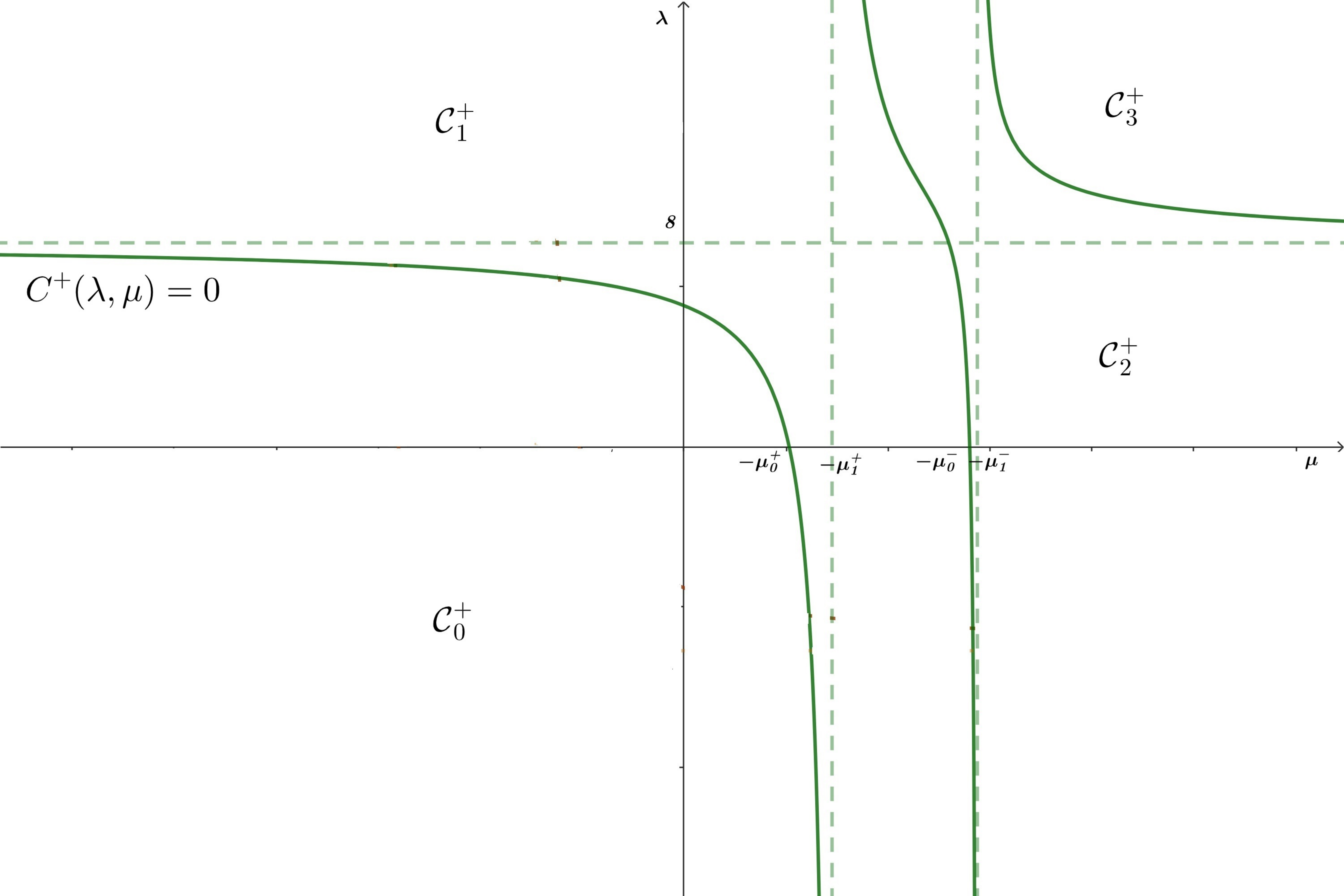}
    \caption{Plot of the curves
defined by the equation $C^+(\lambda,\mu)=0$}
    \label{fig:dynamics2}
\end{figure}

It turns out that in each of the above components $\cC^{-}_{k}$, the
number of eigenvalues of the operator $H_{\lambda\mu}(0)$, lying
below its essential spectrum, remains constant. In a similar way,
any of the components $\cC^{+}_{k}$ is a domain where the number of
eigenvalues of $H_{\lambda\mu}(0)$, lying above the essential
spectrum \eqref{eq:essential_spectrum}, does not vary. Both these
facts are established in the following theorem.

\begin{theorem}\label{teo:constant}
Let $\cC^-$ be one of the above connected components $\cC^-_k$,
$k=0,1,2,3$, of the partition of the $(\lambda,\mu)$-plane. Then for
any $(\lambda,\mu)\in\cC^-$ the number $n_-({H}_{\lambda\mu}(0))$ of
eigenvalues of $H_{\lambda\mu}(0)$ lying below the essential
spectrum $\sigma_\mathrm{ess}\bigl(H_{\lambda\mu}(0)\bigr)$ remains
constant. Analogously, let $\cC^+$ be one of the above connected
components $\cC^+_k$, $k=0,1,2,3$, of the partition of the
$(\lambda,\mu)$-plane. Then for any $(\lambda,\mu)\in\cC^+$ the
number $n_+({H}_{\lambda\mu}(0))$ of eigenvalues of
$H_{\lambda\mu}(0)$ lying above
$\sigma_\mathrm{ess}\bigl(H_{\lambda\mu}(0)\bigr)$ remains constant.
\end{theorem}

The result below concerns the number of eigenvalues of the fiber
Hamiltonian $H_{\lambda\mu}(K)$ for various $K$ and $(\lambda,\mu)$.

\begin{theorem}\label{teo:disc_quasiK}
Let $K\in\T^2$ and $(\lambda,\mu)\in\R^2$. Then for the numbers
$n_+({H}_{\lambda\mu}(K))$ and $n_-({H}_{\lambda\mu}(K))$ of
eigenvalues  of the operator $H_{\lambda\mu}(K)$ lying,
respectively, above and below its essential spectrum
$\sigma_\mathrm{ess}\bigl(H_{\lambda\mu}(K)\bigr)$, the following
two series of implications hold:
\begin{equation}\label{above_eigenK}
\begin{aligned}
& (\lambda,\mu)\in \cC_3^+ \cap \cC_0^- & 
\Longrightarrow \qquad  n_+({H}_{\lambda\mu}(K)) =6, \\
& (\lambda,\mu)\in \cC_2^+ \cap \cC_0^- \,\mathrm{or}
\,(\lambda,\mu)\in \cC_2^+ \cap \cC_1^-  & \hspace*{-14mm}
\Longrightarrow \qquad  n_+({H}_{\lambda\mu}(K)) \ge4,\\
& (\lambda,\mu)\in \cC_1^+ \cap \cC_0^-\,\mathrm{or}
\,(\lambda,\mu)\in \cC_1^+ \cap \cC_1^- &\hspace*{-14mm}
\Longrightarrow \qquad    n_+({H}_{\lambda\mu}(K)) \ge2,\\
& (\lambda,\mu)\in \overline{\cC_0^+}  &  \hspace*{-8mm}
\Longrightarrow \qquad  n_+({H}_{\lambda\mu}(K))\geq 0,
\end{aligned}
\end{equation}
and
\begin{equation}\label{below_eigenK}
\begin{aligned}
& (\lambda,\mu)\in \cC_3^- \cap \cC_0^+ &
\Longrightarrow \qquad n_-({H}_{\lambda\mu}(K)) =6, \\
& (\lambda,\mu)\in \cC_2^- \cap \cC_0^+\,\mathrm{or}
\,(\lambda,\mu)\in \cC_2^- \cap \cC_1^+  & \hspace*{-14mm}
\Longrightarrow \qquad n_-({H}_{\lambda\mu}(K)) \ge4,\\
& (\lambda,\mu)\in \cC_1^- \cap \cC_0^+\,\mathrm{or}
\,(\lambda,\mu)\in \cC_1^- \cap \cC_1^+  & \hspace*{-14mm}
\Longrightarrow \qquad n_-({H}_{\lambda\mu}(K)) \ge2,\\
& (\lambda,\mu)\in \overline{\cC_0^-} & \hspace*{-8mm}
\Longrightarrow \qquad n_-({H}_{\lambda\mu}(K))\geq0,
\end{aligned}
\end{equation}
where  $\overline{\mathcal{A}}$ is the closure of the set
$\mathcal{A}$.
\end{theorem}

The next theorem establishes the exact number of eigenvalues of
$H_{\lambda\mu}(0)$ outside its essential spectrum. In particular,
it shows that the estimates for the numbers
$n_+({H}_{\lambda\mu}(K))$ and $n_-({H}_{\lambda\mu}(K))$ of
eigenvalues of the operator ${H}_{\lambda\mu}(K)$ obtained in
Theorem \ref{teo:disc_quasiK} are sharp.

\begin{theorem}\label{teo:disc_quasi0}
For various $\lambda,\mu\in\R$, the numbers and multiplicities of
eigenvalues of $H_{\lambda\mu}(0)$ outside the set
$\sigma_\mathrm{ess}\bigl(H_{\lambda\mu}(0)\bigr)$  are described in
the following statements.

\begin{itemize}
\item[(\rm{i})] For any $(\lambda,\mu)\in
\cC_{30}=\cC_{3}^-$ the operator $H_{\lambda\mu}(0)$ has exactly
three eigenvalues $z_1(\lambda,\mu;0)$, $z_2(\lambda,\mu;0)$ and
$z_3(\lambda,\mu;0)$ of multiplicity two satisfying
\begin{equation}\label{main01}
z_1(\lambda,\mu;0)<z_2(\lambda,\mu;0) <z_3(\lambda,\mu;0)<0.
\end{equation}

\item[(\rm{ii})] For any $(\lambda,\mu)\in
\cC_{20}:=\cC_{2}^{-}\cap \cC_{0}^{+}$ the operator
$H_{\lambda\mu}(0)$ has two eigenvalues $z_1(\lambda,\mu;0)$ and
$z_2(\lambda,\mu;0)$ of multiplicity two satisfying
\begin{equation}\label{main001}
z_1(\lambda,\mu;0)<z_2(\lambda,\mu;0)<0
\end{equation}
and it has no eigenvalues in $(8, +\infty)$

\item[(\rm{iii})] For any $(\lambda,\mu)\in
\cC_{21}:=\cC_{2}^{-}\cap \cC_{1}^+$, the operator
$H_{\lambda\mu}(0)$ has two eigenvalues $z_1(\lambda,\mu;0)$ and
$z_2(\lambda,\mu;0)$ of multiplicity two  in $(-\infty, 0)$ and it
has one eigenvalue of multiplicity two  in
 $(8, +\infty)$.

\item[(\rm{iv})] For any $(\lambda,\mu)\in
\cC_{11}:=\cC_{1}^{-}\cap \cC_{1}^{+}$, the operator
$H_{\lambda\mu}(0)$ has two eigenvalues $z_1(\lambda,\mu)<0$ and
$z_2(\lambda,\mu)>8$ of multiplicity two .

\item[(\rm{v})] For any $(\lambda,\mu)\in
\cC_{10}:=\cC_{1}^-\cap \cC_{0}^+$, the operator $H_{\lambda\mu}(0)$
has one eigenvalue $z(\lambda,\mu;0)$ of multiplicity two in
$(-\infty, 0)$, nevertheless it has no eigenvalues in $(8,
+\infty)$.

\item[(\rm{vi})] For any $(\lambda,\mu)\in
\cC_{00}:=\cC_{0}^-\cap \cC_{0}^+$, the operator $H_{\lambda\mu}(0)$
has no eigenvalues outside of the essential spectrum.

\item[(\rm{vii})] For any $(\lambda,\mu)\in
\cC_{01}:=\cC_{0}^-\cap \cC_{1}^+$, the operator $H_{\lambda\mu}(0)$
has one eigenvalue $z(\lambda,\mu;0)$ of multiplicity two in $(8,
+\infty)$ and it has no eigenvalues in $(-\infty, 0)$.

\item[(\rm{ix})] For any $(\lambda,\mu)\in
\cC_{02}:=\cC_{0}^-\cap \cC_{2}^+$, the operator $H_{\lambda\mu}(0)$
has two eigenvalues $z_1(\lambda,\mu;0)$ and $z_2(\lambda,\mu;0)$ of
multiplicity two  satisfying
\begin{equation}\label{main001b}
8<z_2(\lambda,\mu;0)<z_1(\lambda,\mu;0)
\end{equation}
and it has no eigenvalues in $(-\infty, 0)$.
\item[(\rm{viii})] For any $(\lambda,\mu)\in
\cC_{12}=\cC_{1}^-\cap \cC_{2}^+$, the operator $H_{\lambda\mu}(0)$
has one eigenvalue of multiplicity two  in $(-\infty, 0)$ and it has
two eigenvalues $z_1(\lambda,\mu;0)$ and $z_2(\lambda,\mu;0)$ of
multiplicity two  in $(8, +\infty)$.

\item[(\rm{x})] For any $(\lambda,\mu)\in
\cC_{03}:=\cC_{3}^+$, the operator $H_{\lambda\mu}(0)$ has exactly
three eigenvalues $z_1(\lambda,\mu;0)$, $z_2(\lambda,\mu;0)$ and
$z_3(\lambda,\mu;0)$ of multiplicity two satisfying
\begin{equation}\label{main01up}
8<z_3(\lambda,\mu;0)<z_2(\lambda,\mu;0)<z_1(\lambda,\mu;0).
\end{equation}
\end{itemize}
\end{theorem}


\begin{figure}\label{1}
 \centering
\includegraphics[width=0.9\textwidth]{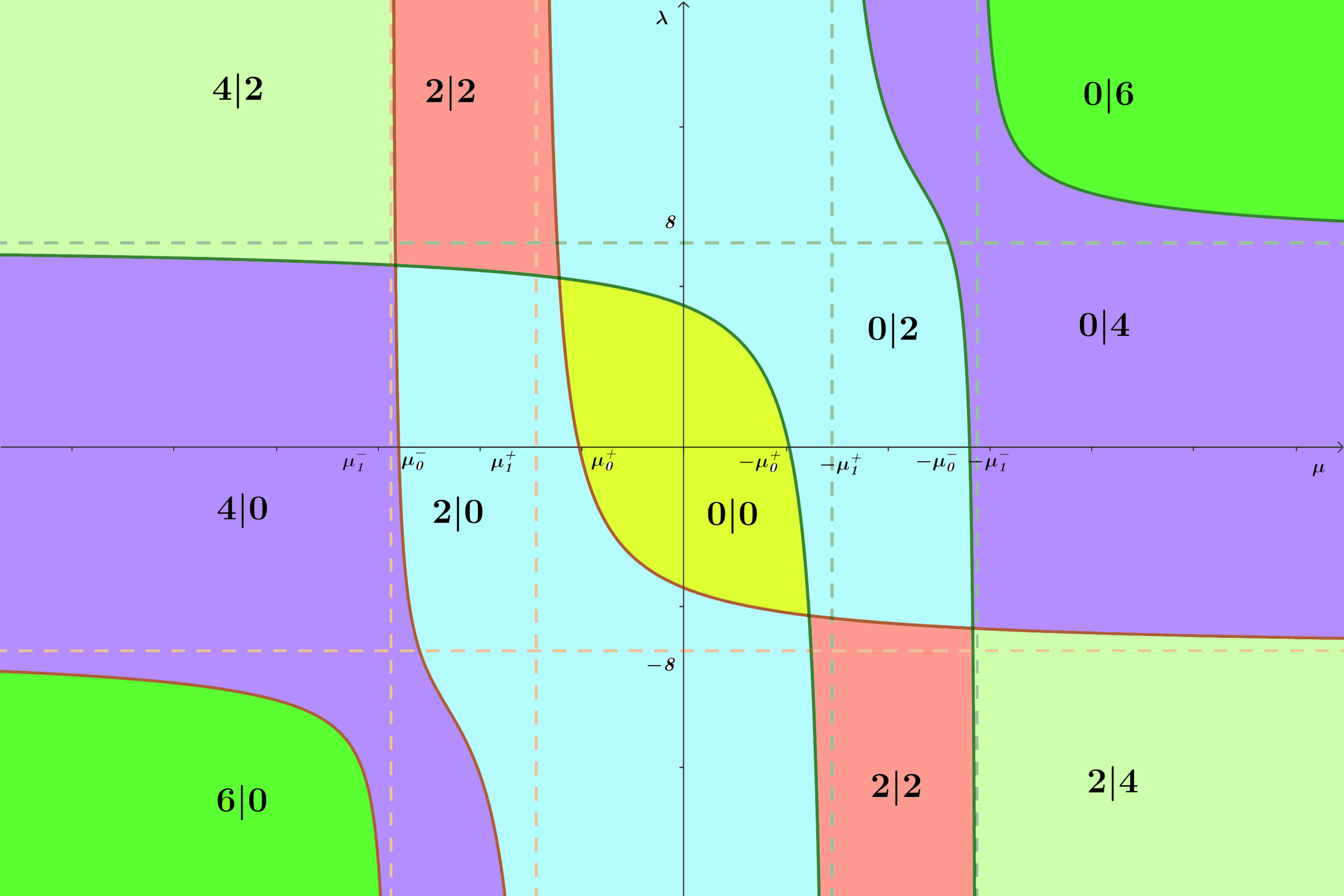}
\caption{ Partition of the $(\lambda,\mu)$-plane of parameters
$\lambda,\mu \in \R$ in the connected components $\cC_{\alpha\beta},
\alpha,\beta=0,1,2,3$ (see Theorem \ref{teo:disc_quasi0}). These
components are tagged by the symbols $N_-|N_+$ formed of the numbers
$N_{-}:=n_-(H_{\lambda, \mu}(0))$ and $N_{+}:=n_+(H_{\lambda,
\mu}(0))$ of eigenvalues of $H_{\lambda\mu}(0)$ lying below and
above the essential spectrum, respectively. Until the point
$(\lambda,\mu)$ does not cross any of the borders between
$\cC_{\alpha\beta}$,  no change occurs in $N_-$ and $N_+$. However,
as soon as $(\lambda,\mu)$ crosses one of those borders, the
essential spectrum of $H_{\lambda\mu}(0)$ either {\it gives birth}
or {\it absorbs} eigenvalues of $H_{\lambda\mu}(0)$.}
\end{figure}

\section{Auxiliary statements}\label{sec:auxiliary}
Let $L^{2,\mathrm{os}}(\T^2)\subset L^{2,\mathrm{o}}(\T^2)$ and
$L^{2,\mathrm{oa}}(\T^2)\subset L^{2,\mathrm{o}}(\T^2)$ be the
subspaces of odd-symmetric and odd-antisymmetric functions defined
as
\begin{align*}
&L^{2,\mathrm{os}}(\T^2)=\{{f}\in L^{2,o}(\T^2): {f}(p_1,p_2)=
{f}(p_2,p_1), \mbox{for a.e}\,\,(p_1,p_2)\in\T^2\}
\end{align*}
and
\begin{align*}
&L^{2,\mathrm{oa}}(\T^2)=\{{f}\in L^{2,o}(\T^2): {f}(p_1,p_2)=
-{f}(p_2,p_1), \mbox{for a.e}\,\,(p_1,p_2)\in\T^2\},
\end{align*}
respectively.

\begin{lemma}
The equality
\begin{equation}\label{evenodd}
L^{2,\mathrm{o}}(\T^2)=L^{2,\mathrm{os}}(\T^2)\oplus
L^{2,\mathrm{oa}}(\T^2)
\end{equation}
holds true.
\end{lemma}
\begin{proof}
The proof follows from the fact that each element of
$L^{2,\mathrm{o}}(\T^2)$ may be represented as the sum of
a function in $L^{2,\mathrm{os}}(\T^2)$ and a function in $L^{2,\mathrm{oa}}(\T^2)$.%
\end{proof}

The operator $H_0(0)$ is the multiplication operator by the
symmetric function $\cE_0(p)=2\epsilon(p)=2\epsilon(p_1,p_2)$ in
$L^{2,o}(\T^2).$ Hence, for each
$\theta\in\{\mathrm{os},\mathrm{oa}\}$ the subspace
$L^{2,\theta}(\T^2)$ is invariant with respect to $H_0(0).$ Now
recall that the interaction operator $V_{\lambda\mu}$ has the form
\eqref{moment_poten} and, thus, it reads as

\begin{align*}
[V_{\lambda\mu}f](p)=&\frac{\lambda}{(2\pi)^2}\int\limits_{\mathbb{T}^2}
(\sin p_1\sin q_1+\sin
p_2\sin q_2) f(q)\mathrm{d}q\\
&+\frac{\mu}{(2\pi)^2}\int\limits_{\mathbb{T}^2} (\sin 2p_1\sin
2q_1+\sin
2p_2\sin 2q_2) f(q)\mathrm{d}q\\
&+\frac{\mu}{2\pi^2}\int\limits_{\mathbb{T}^2}(\sin p_1\cos p_2\sin
q_1\cos q_2 +\sin p_2\cos p_1\sin q_2\cos q_1) f(q)\mathrm{d}q.
\end{align*}
By applying the  equalities
\begin{align*}
2\sin p_1\sin q_1 &+ 2\sin p_2\sin q_2\\
&=(\sin p_1 + \sin p_2)(\sin q_1 + \sin q_2)
 + (\sin p_1 - \sin
p_2)(\sin q_1 - \sin q_2),\\
2\sin 2p_1\sin2&q_1+ 2\sin 2p_2\sin 2q_2\\
&=(\sin 2p_1 + \sin 2p_2)(\sin 2q_1 + \sin 2q_2) + (\sin 2p_1 - \sin
2p_2)(\sin 2q_1 - \sin 2q_2),\\
2\sin p_1\cos p_2&\sin q_1\cos q_2+2\sin p_2\cos p_1\sin q_2\cos q_1\\
&=(\sin p_1\cos p_2 + \sin p_2\cos p_1)(\sin q_1\cos q_2 + \sin q_2\cos q_1)\\
&+(\sin p_1\cos p_2 - \sin p_2\cos p_1)(\sin q_1\cos q_2 - \sin
q_2\cos q_1)
\end{align*}
and denoting by ${V}_{\lambda\mu}^\mathrm{\theta}$   the respective
part of ${V}_{\lambda \mu }$ in the reducing subspace
$L^{2,\theta}(\T^2),\, \theta\in\{\mathrm{os},\mathrm{oa}\}$, we
arrive at the expressions
\begin{align*}
[{V}_{\lambda\mu}^\mathrm{os}{f}](p)=& \frac{\lambda}{8\pi^2}(\sin
p_1 +
\sin p_2)\int\limits_{\mathbb{T}^2} (\sin q_1 + \sin q_2) f(q)\mathrm{d}q\\
&+\frac{\mu}{8\pi^2}(\sin 2p_1 +
\sin 2p_2)\int\limits_{\mathbb{T}^2} (\sin 2q_1 + \sin 2q_2) f(q)\mathrm{d}q\\
&+\frac{\mu}{4\pi^2}(\sin p_1\cos p_2 + \sin p_2\cos
p_1)\int\limits_{\mathbb{T}^2} (\sin q_1\cos q_2 + \sin q_2\cos q_1)
f(q)\mathrm{d}q
\end{align*}
and
\begin{align*}
[{V}_{\lambda\mu}^\mathrm{oa}{f}](p)=& \frac{\lambda}{8\pi^2}(\sin
p_1 -
\sin p_2)\int\limits_{\mathbb{T}^2} (\sin q_1 - \sin q_2) f(q)\mathrm{d}q\\
&+\frac{\mu}{8\pi^2}(\sin 2p_1 -
\sin 2p_2)\int\limits_{\mathbb{T}^2} (\sin 2q_1 - \sin 2q_2) f(q)\mathrm{d}q\\
&+\frac{\mu}{4\pi^2}(\sin p_1\cos p_2 - \sin p_2\cos
p_1)\int\limits_{\mathbb{T}^2} (\sin q_1\cos q_2 - \sin q_2\cos q_1)
f(q)\mathrm{d}q.
\end{align*}
It follows from the above expressions that
$$
{H}_{\lambda\mu}(0)\big|_{L^{2,\theta}(\T^2)} =
{H}_{\lambda\mu}^\theta(0):=H_0 + {V}_{\lambda\mu}^\theta,\,\,
\text{for}\,\, \theta\in\{\mathrm{os},\mathrm{oa}\}.
$$
Therefore,
\begin{equation}\label{sigma_Hmu}
\sigma\Big({H}_{\lambda\mu}(0)\Big)=
\bigcup_{\theta\in\{\mathrm{os},\mathrm{oa}\}}
\sigma\Big({H}_{\lambda\mu}(0)\big|_{L^{2,\theta}(\T^2)}\Big),
\end{equation}
\black where $A\big|_\cH$ is the restriction of a self-adjoint
operator $A$ on a reducing subspace $\cH$. Thus, the study of the
discrete spectrum of  ${H}_{\lambda\mu}(0)$ reduces to that for the
restrictions of ${H}_{\lambda\mu}(0)$ onto each subspace
$L^{2,\theta}(\T^2)$, $\theta\in\{\mathrm{os},\mathrm{oa}\}$.

\subsection{The Lippmann--Schwinger operator}

Let $\{\alpha_{i}^\theta,\,\,i=1,2,3\} $ be a  system of vectors in
$L^{2,\theta}(\T^2),\,\theta\in\{\mathrm{os},\mathrm{oa}\}$, with
\begin{equation}\label{ons1}
\alpha_{1}^{os}(p)=\dfrac{\sin p_{1}+\sin
p_{2}}{2\pi},\,\,\alpha_{2}^{os}(p)=\dfrac{\sin 2p_{1}+\sin
2p_{2}}{2\pi},\,\,\alpha_{3}^{os}(p)=\dfrac{\sin p_{1}\cos
p_{2}+\sin p_{2}\cos p_{1} }{\sqrt{2}\pi}
\end{equation}
and
\begin{equation}\label{ons2}
\alpha_{1}^{oa}(p)=\dfrac{\sin p_{1}-\sin
p_{2}}{2\pi},\,\,\alpha_{2}^{oa}(p)=\dfrac{\sin 2p_{1}-\sin
2p_{2}}{2\pi},\,\,\alpha_{3}^{oa}(p)=\dfrac{\sin p_{1}\cos
p_{2}-\sin p_{2}\cos p_{1} }{\sqrt{2}\pi}.
\end{equation}
One easily verifies by inspection that the vectors \eqref{ons1} and
\eqref{ons2} are orthonormal in $L^{2,o}(\T^2)$. By using the
orthonormal systems \eqref{ons1} and \eqref{ons2}  one obtains
\begin{align}\label{repr1}
&{V}_{\lambda\mu}^\theta{f}=\frac{\lambda}{2}({f},\alpha_{1}^\theta)\alpha_{1}^\theta+
\frac{\mu}{2}({f},\alpha_{2}^\theta)\alpha_{2}^\theta+
\frac{\mu}{2}({f},\alpha_{3}^\theta)\alpha_{3}^\theta ,\quad
\theta\in\{\mathrm{os},\,\,\mathrm{oa}\},
\end{align}
where  $(\cdot,\cdot)$ is the inner product in $L^{2,\theta}(\T^2).$
For any $z\in\C \setminus[0,\,8]$ we define (the transpose of) the
Lippmann-Schwinger operator (see., e.g., \cite{LSchwinger:1950}) as
\begin{equation*}\label{LSchw1}
{B}_{\lambda\mu}^{\theta}(0,z)=-{V}_{\lambda\mu}^{\theta}{R}_0(0,z),\,
\theta\in\{\mathrm{os},\mathrm{oa}\},
\end{equation*}
where  ${R}_0(0,z):= [{H}_0(0)-zI]^{-1},\,
z\in\mathbb{C}\setminus[0,\,8]$, is the resolvent of the operator
${H}_0(0)$.
\begin{lemma}\label{eigen-eigenvalue}
For each $\lambda,\mu \in \R$  the number $z\in \mathbb{C}\setminus
[0,\,8]$ is an eigenvalue of the operator
${H}_{\lambda\mu}^{\theta}(0)$ if and only if the number $1$ is an
eigenvalue for
 ${B}_{\lambda\mu}^{\theta}(0,z),\,\theta\in\{\mathrm{os},\mathrm{oa}\}$.
\end{lemma}
The proof of this lemma is quite standard (see., e.g.,
\cite{Albeverio:1988}) and, thus, we omit~it.

In the following, we identify the symbols $os$ and $oa$ with the
signs $\mathrm{+}$ and $\mathrm{-}$, respectively.

The representation \eqref{repr1} yields the equivalence of the
Lippmann-Schwinger equation
\begin{align*}
{B}_{\lambda\mu}^{\theta}(0,z){\varphi}={\varphi}, \,\,\,{\varphi}
\in L^{2,\theta}(\T^2),\,\theta\in\{\mathrm{os},\mathrm{oa}\}
\end{align*}
to the following algebraic linear system in
$x_i:=({\varphi},\alpha_{i}^\theta),\,\,i=1,2,3$:
\begin{equation}\label{system}
\left\lbrace\begin{array}{ccc} (1+\lambda
a^{\pm}(z))x_{1}+ \lambda c^{\pm}(z) x_{2}+\lambda  d^{\pm}(z)x_{3}=0,\\
\mu  c^{\pm}(z)x_{1}+(1+\mu b^{\pm}(z))x_{2}+\mu  e^{\pm}(z)x_{3}=0,\\
2\mu  d^{\pm}(z)x_{1}+2\mu e^{\pm}(z)x_{2}+(1+\mu f^{\pm}(z))x_{3}=0,\\
\end{array}\right.
\end{equation}
where
\begin{align}
\label{fa} &a^{\pm}(z) =\frac{1}{8\pi^2}\int\limits_
{\mathbb{T}^2}\frac{(\sin p_{1}\pm\sin p_2)^2\ \mathrm{d}p}{\cE_{0}(p)-z},\\
\label{fb} &b^{\pm}(z) =\frac{1}{8\pi^2}\int\limits_
{\mathbb{T}^2}\frac{(\sin 2p_{1}\pm\sin 2p_2)^2\ \mathrm{d}p}{\cE_{0}(p)-z},\\
\label{fc} &c^{\pm}(z) =\frac{1}{8\pi^2}\int\limits_
{\mathbb{T}^2}\frac{(\sin p_{1}\pm\sin p_2)(\sin 2p_{1}\pm\sin 2p_2)
\mathrm{d}p}{\cE_{0}(p)-z},\\
\label{fd} &d^{\pm}(z) =\frac{\sqrt{2}}{8\pi^2}\int\limits_
{\mathbb{T}^2}\frac{(\sin p_{1}\pm\sin p_2)
(\sin p_{1}\cos p_2\pm\sin p_2\cos p_1) \mathrm{d}p}{\cE_{0}(p)-z},\\
\label{fe} &e^{\pm}(z) =\frac{\sqrt{2}}{8\pi^2}\int\limits_
{\mathbb{T}^2}\frac{(\sin 2p_{1}\pm\sin 2p_2)
(\sin p_{1}\cos p_2\pm\sin p_2\cos p_1) \mathrm{d}p}{\cE_{0}(p)-z},\\
\label{ff} &f^{\pm}(z) =\frac{1}{2\pi^2}\int\limits_
{\mathbb{T}^2}\frac{(\sin p_{1}\cos p_2\pm\sin p_2\cos p_1)^2
\mathrm{d}p}{\cE_{0}(p)-z}.
\end{align}
It is easy to check that the functions
$a^{\pm}(z),b^{\pm}(z),c^{\pm}(z),d^{\pm} (z) , e^ {\pm}(z)$ and
$f^{\pm}(z)$ do not depend on the sign $\pm$. Thus,  we skip the
sign superscripts and denote these functions simply by
$a(z),b(z),c(z),d(z), e(z)$ and $f(z)$.

\begin{remark}
\label{Rabcd} One easily verifies by inspection that the following
relations hold:
\begin{align*}
&b(z)+\sqrt{2}e(z)=(4-z)\,c(z),\\
&\sqrt{2}e(z)+f(z)=\sqrt{2}\,(4-z)\,d(z),\\
&c(z)+\sqrt{2}\,d(z)=(4-z)\,a(z)+\frac{1}{2}.
\end{align*}
\end{remark}

The reasoning similar to the one we used in the case of the
functions \eqref{fa}--\eqref{ff} allows us to conclude that the
determinant of the operator $I-{B}_{\lambda\mu}^{\theta}(0,z)$  does
not depend on $\theta$, too. Thus we write
$$\Delta_{\lambda\mu}(z):=
\Delta_{\lambda\mu}^{\theta}(z):=\det[I-{B}_{\lambda\mu}^{\theta}(0,z)],
\,\,z\in\mathbb{C}\setminus[0,\,8],\,\theta\in\{\mathrm{os},\mathrm{oa}\}.$$

\begin{lemma}\label{eigen-zeros}
A number $z\in \mathbb{C}\setminus [0,\,8]$ is an eigenvalue of the
operator ${H}_{\lambda\mu}(0)$ if and only if
\begin{equation}\label{cont01}
\Delta_{\lambda\mu}(z)=0.
\end{equation}
If $z$ is such an eigenvalue then necessarily it has multiplicity
two.
\end{lemma}

The proof of this lemma for each
$\theta\in\{\mathrm{os},\mathrm{oa}\}$ is quite standard (cf., e.g.,
\cite{LKhKh:2021,LBozorov:2009})

\begin{remark}
Notice that for each $\theta\in\{\mathrm{os},\mathrm{oa}\}$ any
eigenvalue $z\in \mathbb{C}\setminus [0,\,8]$ of
${H}^{\theta}_{\lambda\mu}(0)$ is simple. Since
${H}_{\lambda\mu}(0)={H}^{os}_{\lambda\mu}(0)\oplus{H}^{oa}_{\lambda\mu}(0)$,
the same $z$ is an eigenvalue of multiplicity two for
${H}_{\lambda\mu}(0)$.
\end{remark}

\begin{lemma}\label{determ_form}
For any $\lambda,\mu \in \R$  the determinant
$\Delta_{\lambda\mu}(z)$ has the form
\begin{align}\label{determ}
\Delta_{\lambda\mu}(z)=\Delta_{\lambda0}(z)\Delta_{0\mu}(z)+\Delta^{(12)}_{\lambda\mu}(z),
\end{align}
where
\begin{align}\label{determ1}
&\Delta_{\lambda0}(z)=1+\lambda a(z),\\ \label{determ2}
&\Delta_{0\mu}(z)=(1+\mu b(z))(1+\mu f(z))-2\mu^2e^2(z),\black\\
\label{determ3}
&\Delta^{(12)}_{\lambda\mu}(z)=4\lambda\mu^2c(z)d(z)e(z)-\lambda\mu
c^2(z)(1+\mu f(z)) -2\lambda\mu d^2(z)(1+\mu b(z)).
\end{align}
\end{lemma}
\begin{proof}
Direct computation of the determinant gives the result.
\end{proof}

\begin{lemma}\label{asympt(abcdef)}
The functions $a(z), b(z),c(z), d(z),e(z)$ and $f(z)$ defined in
$\R\setminus[0,8]$ are real-valued and, moreover, strictly
increasing and  positive in $(-\infty, 0)$, strictly increasing and
negative in $(8,+\infty)$ and have following asymptotics:
\begin{align}\label{asimp_function0}
&\lim\limits_{z\nearrow0}a(z)=\frac{\pi-2}{2\pi},\quad
\lim\limits_{z\nearrow   0}b(z)=\frac{30\pi-92}{3\pi},
\quad \lim\limits_{z\nearrow   0}c(z)=\frac{2\pi-6}{\pi},\\
\label{asimp_functions8} &\lim\limits_{z\nearrow
0}d(z)=\frac{4-\pi}{2\sqrt{2}\pi},\quad \lim\limits_{z\nearrow
0}e(z)=\frac{20-6\pi}{3\sqrt{2}\pi},\quad \lim\limits_{z\nearrow
0}f(z)=\frac{4}{3\pi}
\end{align}
and
\begin{align}
\label{asimp_functions3} &\lim\limits_{z\searrow
8}a(z)=\frac{2-\pi}{2\pi},\quad \lim\limits_{z\searrow
8}b(z)=\frac{92-30\pi}{3\pi},
\quad\lim\limits_{z\searrow 8}c(z)=\frac{6-2\pi}{\pi},\\
\label{asimp_functions4} &\lim\limits_{z\searrow
8}d(z)=\frac{\pi-4}{2\sqrt{2}\pi},\quad \lim\limits_{z\searrow
8}e(z)=\frac{6\pi-20}{3\sqrt{2}\pi},\quad \lim\limits_{z\searrow
8}f(z)=-\frac{4}{3\pi}.
\end{align}
\end{lemma}
The proof of Lemma \ref{asympt(abcdef)} is given in Appendix
\ref{sec:append_A}. \hfill $\square$

\begin{lemma}\label{lemm:asympt}
The function $\Delta_{\lambda\mu}(z)$ is real-valued on
$\R\setminus[0,8]$ and has the following asymptotics:
\begin{align}\label{asymp_determ}
&\Delta_{\lambda\mu}(z)=C^{-}(\lambda,\mu)+o(1),\,\,\text{as}\,\,
z\nearrow
0, \\
&\Delta_{\lambda\mu}(z)=C^{+}(\lambda,\mu)+o(1),\,\,\text{as}
\,\,z\searrow 8,
\end{align}
where the functions $C^\pm(\lambda,\mu)$ are given by
\eqref{abobel}.
\end{lemma}
\begin{proof}
In view of Lemmas \ref{determ_form} and \ref{asympt(abcdef)}, the
proof is obtained by an immediate inspection taking into account the
definitions \eqref{root01} and \eqref{root02}.
\end{proof}

\black Next we study the number and location of roots of the
functions $\Delta_{\lambda0}$ and $\Delta_{0\mu}$ defined by
\eqref{determ1} and \eqref{determ2}.

\begin{lemma}\label{simple1} Let $\lambda\in\mathbb{R}.$
\begin{itemize}
\item[(\rm{i})]If
$\lambda<\frac{2\pi}{2-\pi},$ then $\Delta_{\lambda0}(\cdot)$ has a
unique root $\zeta^-(\lambda,0)$ in  $(-\infty,0)$, and it has no
roots in  $(8,+\infty)$.

\item[(\rm{ii})] If $\lambda\in[\frac{2\pi}{2-\pi},
\frac{2\pi}{\pi-2}],$ then $\Delta_{\lambda0}(\cdot)$  has no roots
in $\R\setminus [0,8]$.

\item[(\rm{iii})] If $\lambda>\frac{2\pi}{\pi-2},$
then $\Delta_{\lambda0}(\cdot)$  has a unique root
$\zeta^+(\lambda,0)$ in $(8,+\infty)$, and it has no roots in
$(-\infty, 0)$.
\end{itemize}
\end{lemma}

\begin{lemma}\label{simple2}
Let $\mu \in \R$  and let $\mu_0^-$ be as in \eqref{root01}. Then:
\begin{itemize}
\item[(\rm{i})] For any $\mu<\mu_{0}^-$ the functions
$1+\mu b(z)$ and $1+\mu f(z)$ have unique roots $\eta^{(-)}_1(\mu)$
and $\eta^{(-)}_2(\mu)$  in $(-\infty,0),$ respectively.

\item[(\rm{ii})] For any $\mu>-\mu_{0}^-$
the functions $1+\mu b(z)$ and $1+\mu f(z)$ have unique roots
$\eta^{(+)}_1(\mu)$ and $\eta^{(+)}_2(\mu)$  in $(8,+\infty),$
respectively.
\end{itemize}
\end{lemma}
\begin{proof}
Both Lemmas \ref{simple1} and \ref{simple2} are proven by using the
representations \eqref{determ1} and \eqref{determ2} from Lemma
\ref{determ_form} and the asymptotical
formulae for $a(z)$, $b(z)$ and $f(z)$ from Lemma \ref{asympt(abcdef)}.%
\end{proof}

Let $\eta^{(-)}_1(\mu)$, $\eta^{(-)}_2(\mu)$ and
$\eta^{(+)}_1(\mu)$, $ \eta^{(+)}_2(\mu)$ be  two different pairs of
roots of $\Delta_{0\mu}$ such that
$\eta^{(-)}_1(\mu),\,\eta^{(-)}_2(\mu)\in (-\infty,0)$ and
$\eta^{(+)}_1(\mu),  \eta^{(+)}_2(\mu)\in(8,+\infty )$. We set:
\begin{align*}
&\eta^{-}_{\min}(\mu)=\min\{\eta^{(-)}_1(\mu),\eta^{(-)}_2(\mu)\},\\
&\eta^{-}_{\max}(\mu)=\max\{\eta^{(-)}_1(\mu),\eta^{(-)}_2(\mu)\}.
\end{align*}
and
\begin{align*}
&\eta^{+}_{\min}(\mu)=\min\{\eta^{(+)}_1(\mu),\eta^{(+)}_2(\mu)\},\\
&\eta^{+}_{\max}(\mu)=\max\{\eta^{(+)}_1(\mu),\eta^{(+)}_2(\mu)\}.
\end{align*}

The next lemma describes the dependence of the number of roots of
the function $\Delta_{0\mu}$ in $\mathbb{R}\setminus[0,8]$ and their
location on the magnitude $\mu\in\R$.

\begin{lemma}\label{simple3} Let $\mu\in\mathbb{R}$ and
let the numbers $\mu_0^-$, $\mu_0^+$ be as in \eqref{root01} and
satisfy inequalities \eqref{roots}.

\begin{itemize}
\item[(\rm{i})] If $\mu<\mu_{0}^-$ then
$\Delta_{0\mu}(\cdot)$  has exactly two roots $\zeta_{1}^-(0,\mu)$
and $\zeta_2^-(0,\mu)$ satisfying
\begin{equation}\label{main310}
\zeta_1^-(0,\mu)<\eta_{\min}^-(\mu)\leq\eta_{\max}^-(\mu)<\zeta_2^-(0,\mu)
<0.
\end{equation}

\item[(\rm{ii})]
If $\mu\in[\mu_{0}^-,\mu_{0}^+)$ then $\Delta_{0\mu}(\cdot)$  has a
unique root $\zeta^-(0,\mu)$ in the interval $(-\infty,0)$, and  it
has no roots in the interval $(8,+\infty)$.

\item[(\rm{iii})] If $\mu\in[\mu_{0}^+,-\mu_{0}^+]$ then
$\Delta_{0\mu}(\cdot)$  has no  roots in $\R\setminus [0,8]$.

\item[(\rm{iv})]
If $\mu\in(-\mu_{0}^+,-\mu_{0}^-]$ then $\Delta_{0\mu}(\cdot)$  has
a unique root $\zeta^+(0,\mu)$ in the interval $(8,+\infty)$, and it
has no roots in the interval $(-\infty,0)$.
\item[(\rm{v})]
\underline{} If $\mu>-\mu_{0}^-$ then $\Delta_{0\mu}(\cdot)$ has
exactly two roots $\zeta_1^+(0,\mu)$ and $\zeta_2^+(0,\mu)$
satisfying
\begin{equation}\label{main3.11}
8<\zeta_2^+(0,\mu)<\eta_{\min}^+(\mu)\leq\eta_{\max}^+(\mu)<\zeta_1^+(0,\mu).
\end{equation}
\end{itemize}
\end{lemma}
\black

\begin{proof} Let us prove the item (i).
Combining the hypothesis of this item with \eqref{roots}  yields

\begin{align*}
&C^{-}(0,\mu)=\frac{120\pi-12\pi^2-256}{3\pi^2}(\mu-\mu_0^+)(\mu-\mu_0^-)
>0.
\end{align*}
Therefore by Lemma \ref{lemm:asympt}
\begin{align*}
&\lim\limits_{z\nearrow 0} \Delta_{0\mu}(z)>0.
\end{align*}
Moreover
\begin{equation*}
\lim\limits_{z\to -\infty }\Delta_{0\mu}(z)=1.
\end{equation*}
Lemma \ref{simple2} implies that
\begin{align*}
&\Delta_{0\mu}(\eta^{-}_{\min}(\mu))<0
\end{align*}
and
\begin{align*}
&\Delta_{0\mu}(\eta^{-}_{\max}(\mu))<0.
\end{align*}
Hence there exists two numbers
$\zeta_1^-(0,\mu)\in(-\infty,\eta_{\min}(\mu))$ and
$\zeta_2^-(0,\mu)\in(\eta_{\max}(\mu),0)$ such that
\begin{align*}
&\Delta_{0\mu}(\zeta_1^-(0,\mu))=0\,\,\, \mathrm{and} \,\,\,
\Delta_{0\mu}(\zeta_2^-(0,\mu))=0.
\end{align*}
Since ${V}_{0\mu}^{\theta},\,\theta\in\{\mathrm{os},\mathrm{oa}\}$
is operator of rank two the determinant $\Delta_{0\mu}(z)$ has no
more than two roots in $\R\setminus [0,8]$, which completes the
proof of item (i).

The remaining  items are proven in a similar way.
\end{proof}

\section{Proofs of the main  results}\label{sec:proofs}

In this section we prove our main results, Theorems
\ref{teo:discr_Kne0} --\ref{teo:disc_quasi0}. We will see below that
Theorems \ref{teo:discr_Kne0} and \ref{teo:disc_quasiK}  are rather
corollaries of Theorems \ref{teo:constant} and
\ref{teo:disc_quasi0}. Thus, we start with a proof of the two last
ones.
\smallskip

\textit{Proof of Theorem \ref{teo:constant}}. Since for any
$(\lambda,\mu)\in\cC$ the determinant $\Delta_{\lambda\mu}(z)$ is
real analytic in $z \in \{\Re \,z<0\}$ and the equalities
\begin{align*}
&\lim\limits_{z\to-\infty} \Delta_{\lambda\mu}(z)=1,\quad
\lim\limits_{z\nearrow 0} \Delta_{\lambda\mu}(z)<0
\end{align*} are hold there exist negative numbers $B_1<B_2<0$ such
that the function has only finite number roots in $(B_1,B_2)$.

Let $(\lambda_0,\mu_0)$ be a point of $\cC$ and $z_0<0$ be a root of
multiplicity $m\geq 1$ of the function $\Delta_{\lambda_0\mu_0}(z)$.
For each fixed $z<0$ the determinant $\Delta_{\lambda\mu}(z)$ is a
real analytic function in $(\lambda,\mu)\in\cC$ and for each
$\lambda,\mu\in \R$ the function $\Delta_{\lambda\mu}(z)$ is real
analytic in $z\in(-\infty,0)$. Hence, for each $\varepsilon>0$ there
are numbers $\delta>0$, $\eta> 0$ and an open neighborhood
$W_{\eta}(z_0)$ of $z_0$ with radius $\eta$ such that for all $z\in
\overline{W_{\eta}(z_0)}$ and $(\lambda,\mu )\in\cC$ obeying the
conditions $|z-z_0|=\eta$ and
$||(\lambda,\mu)-(\lambda_0,\mu_0)||<\delta$ the following two
inequalities $|\Delta_{\lambda_0\mu_0}(z)|>\eta$ and
$|\Delta_{\lambda\mu}(z)-\Delta_{\lambda_0\mu_0}(z)|<\epsilon $
hold. Then by Rouch\'e's theorem the number of roots of the function
$\Delta_{\lambda\mu}(z)$ in ${W_{\eta}(z_0)}$ remains constant  for
all $(\lambda,\mu)\in\cC$ satisfying
$||(\lambda,\mu)-(\lambda_0,\mu_0)||<\delta $. Since the root
$z_0<0$ of the function $\Delta_{\lambda\mu}(z)$ is arbitrary in
$(B_1,B_2)$ we conclude that the number of its roots  remains
constant in $(B_1,B_2)$ for all $(\lambda,\mu)\in\cC$ satisfying
$||(\lambda,\mu)-(\lambda_0,\mu_0)||<\delta$.

Further each Jordan curve $\gamma\subset\cC$ connecting any two
points of $\cC$ is a {\it compact set}, so the number of roots of
the function $\Delta_{\lambda\mu}(z)$ lying below zero for any
$(\lambda,\mu)\in \gamma$ remains constant. Therefore, Lemma
\ref{eigen-zeros} yields that the number of eigenvalues
$n_-\bigl(H_{\lambda\mu}(0)\bigr)$ of the operator
${H}_{\lambda\mu}(0)$ below the essential spectrum  is constant.

The proof in the case of $n_+\bigl(H_{\lambda\mu}(0)\bigr)$ is done
in the same way. \hfill $\square$
\smallskip

\noindent\textit{Proof of Theorem \ref{teo:disc_quasi0}}. We only
prove items (i), (ii) and (v). The remaining items can be proven
similarly.

(\rm{i}). Assume that $(\lambda,\mu)\in \cC_{30}=\cC_{3}^-$. Lemma
\ref{simple3} yields that for $\mu<\mu_0^-$ the function
$\Delta_{0\mu}(z)$ has exactly two roots $\zeta_{1}^-(0,\mu)$ and
$\zeta_2^-(0,\mu)$ satisfying the relation \eqref{main310}.

Since $\mu<0$, the functions $1+\mu b(z)$ and $1+\mu f(z)$ are
continuous and monotonously decreasing in $(-\infty,0)$. Obviously,
by Lemma \ref{simple2}\,(i) we have
\begin{align*}
&1+\mu b(z)>0,\,1+\mu
f(z)>0\,\,\mathrm{if}\,\,z<\eta_{\min}^{-}(\mu),\\
&1+\mu b(z)<0,\,1+\mu
f(z)<0\,\,\mathrm{if}\,\,z>\eta_{\max}^{-}(\mu).
\end{align*}
For the roots $\zeta_{1}^-(0,\mu)$ and $\zeta_2^-(0,\mu)$ of the
equation $\Delta_{0\mu}(z)=(1+\mu b(z))(1+\mu f(z))-2\mu^2e^2(z)=0$
we then have

\begin{equation}\label{equalityroot1}
\sqrt{1+\mu f(\zeta_1^-(0,\mu))}\sqrt{1+\mu
b(\zeta_1^-(0,\mu))}=-\sqrt{2}\mu e(\zeta_1^-(0,\mu))
\end{equation}
and
\begin{equation}\label{equalityroot2}
\sqrt{-(1+\mu f(\zeta_2^-(0,\mu)))}\sqrt{-(1+\mu
b(\zeta_2^-(0,\mu)))}=-\sqrt{2}\mu e(\zeta_2^-(0,\mu)).
\end{equation}

In view of $C^{-}(\lambda,\mu)<0$, Lemma \ref{lemm:asympt} implies
\begin{align*}
&\lim\limits_{z\nearrow 0}
\Delta_{\lambda\mu}(z)<0\,\,\mathrm{and}\,\,\lim\limits_{z\to-\infty}
\Delta_{\lambda\mu}(z)  =1.
\end{align*}
Using the explicit representation \eqref{determ} for
$\Delta_{\lambda\mu}$ and the identity \eqref{equalityroot1} one
arrives with the following chain of equalities:
\begin{align}\label{good1}
&\Delta_{\lambda\mu}(\zeta_1^-(0,\mu)) \\
&=\Delta_{\lambda0}(\zeta_1^-(0,\mu))\Delta_{0\mu}(\zeta_1^-(0,\mu))
+\Delta^{(12)}_{\lambda\mu}(\zeta_1^-(0,\mu))\nonumber \\
&=4\lambda\mu^2c(\zeta_1^-(0,\mu))d(\zeta_1^-(0,\mu))e(\zeta_1^-(0,\mu)) \nonumber\\
&-\lambda\mu c^2(\zeta_1^-(0,\mu))(1+\mu f(\zeta_1^-(0,\mu)))\nonumber \\
&-2\lambda\mu d^2(\zeta_1^-(0,\mu))(1+\mu b(\zeta_1^-(0,\mu)))\nonumber\\
&=-\lambda\mu \left(c(\zeta_1^-(0,\mu))\sqrt{1+\mu
f(\zeta_1^-(0,\mu))}+d(\zeta_1^-(0,\mu))\sqrt{2(1+\mu
b(\zeta_1^-(0,\mu)))}\right)^2 \nonumber.
\end{align}

Clearly, \eqref{good1} implies
\begin{equation}\label{ineq1}
\Delta_{\lambda\mu}(\zeta_1^-(0,\mu))<0.
\end{equation}
Analogously, using the identity \eqref{equalityroot2} one obtains
\begin{align}\label{ineq2}
&\Delta_{\lambda\mu}(\zeta_2^-(0,\mu))>0.
\end{align}
Inequalities \eqref{ineq1} and \eqref{ineq2} together with the
relations

\begin{align*}
&\lim\limits_{z\to-\infty} \Delta_{\lambda\mu}(z)  =1,
\lim\limits_{z\nearrow 0} \Delta_{\lambda\mu}(z)<0.
\end{align*}
yield the existence three roots $z_1$, $z_2$ and $z_3$ of function
$\Delta_{\lambda\mu}(z).$

Hence the function $\Delta_{\lambda\mu}(z)$ has three single roots
smaller than $0$. From Lemma \ref{eigen-zeros} it follows that the
operator $H_{\lambda\mu}(0)$ has six eigenvalues (counting
multiplicities) below the essential spectrum. Since the interaction
operator $V_{\lambda\mu}$ has rank at most six, $H_{\lambda\mu}(0)$
has no eigenvalues above its essential spectrum.

(\rm{ii}) The hypothesis of the theorem implies that
$(\lambda,\mu)\in \cC_{20}=\cC_{2}^{-}\cap \cC_{0}^{+}$, and, hence
$\mu<\mu_1^+$, i.e., $\mu<\mu_0^-$ or $\mu_0^-\leq\mu<\mu_0^+$, and
hence the items (i) and (ii) of Lemma \ref{simple3} yield that the
function $\Delta_{0\mu}(z)$ has at least one root in the interval
$(-\infty,0)$, which we denote by $\zeta^-(0,\mu)$.

It can be shown as \eqref{ineq1}  the inequality
$\Delta_{\lambda\mu}(\zeta^-(0,\mu))<0$ holds.

The definitions of the sets  $\cC_{2}^{-}$  yields the inequalities
\begin{align*}
C^{-}(\lambda,\mu)>0
\end{align*}
and hence Lemma \ref{lemm:asympt} gives
\begin{align*}
&\lim\limits_{z\nearrow 0} \Delta_{\lambda\mu}(z)>0.
\end{align*}
By definition of $\Delta_{\lambda\mu}(z)$ we have
\begin{align*}
\lim\limits_{z\to\pm\infty} \Delta_{\lambda\mu}(z)  =1.
\end{align*}

The relations above obtained relations yields
\begin{align*}
&\lim\limits_{z\to-\infty} \Delta_{\lambda\mu}(z)
=1,\,\,\,\Delta_{\lambda\mu}(\zeta^-(0,\mu))<0,\,\,\,\lim\limits_{z\nearrow
0} \Delta_{\lambda\mu}(z)>0
\end{align*}
and hence the function $\Delta_{\lambda\mu}(z)$ has only two roots
$z_1$ and $z_2$  of satisfying
\begin{equation}\label{main001c}
z_1(\lambda,\mu;0)<\zeta^-(0,\mu)<z_2(\lambda,\mu;0)<0.
\end{equation}
Otherwise it would have at least four roots in $(-\infty,0)$, but
this is impossible.

\black Now we show that the function has no roots in $(8,+\infty)$.

By definitions of $\cC^{+}$ and $\Delta_{\lambda\mu}(z)$ we have
$C^{+}(\lambda,\mu)>0$ and
\begin{align*}
&\lim\limits_{z\searrow 8} \Delta_{\lambda\mu}(z)>0,\quad
\lim\limits_{z\to+\infty} \Delta_{\lambda\mu}(z)  =1,
\end{align*}
which yields that $\Delta_{\lambda\mu}(z)$ has no roots in
$(8,+\infty).$ Otherwise it would have at least two roots in
$(8,+\infty)$.

Hence, Lemma \ref{eigen-zeros} implies that the operator
$H_{\lambda\mu}^\theta,\,\theta\in\{\mathrm{os},\mathrm{oa}\}$ has
no eigenvalues above the essential spectrum.

 (\rm{v}) Assume $(\lambda,\mu)\in \cC_{10}=\cC_{1}^-\cap
\cC_{0}^+$. By Lemma \ref{simple1}, for any
$\lambda<\frac{2\pi}{2-\pi}$ the operator $H_{\lambda0}^\theta$ has
unique eigenvalue in $(-\infty,0)$ at the point $(\lambda,0)\in
\cC_{10}$. Then by Theorem \ref{teo:constant} for any
$(\lambda,\mu)\in \cC_{10}$ the operator
$H^\theta_{\lambda\mu}(0),\,\theta\in\{\mathrm{os},\mathrm{oa}\}$
has unique eigenvalue  in $(-\infty,0)$.

\black
\subsection{The discrete spectrum of ${H}_{\lambda\mu}(K)$}

For every $n\ge1$ define
\begin{equation}\label{enK}
e_n(K;\lambda,\mu):= \sup\limits_{\phi_1,\ldots,\phi_{n-1}\in
L^{2,o}(\T^2)}\,\,\inf\limits_{{\psi}
\in[\phi_1,\ldots,\phi_{n-1}]^\perp,\,\|{\psi}\|=1}
({H}_{\lambda\mu}(K){\psi},{\psi})
\end{equation}
and
\begin{equation}\label{EnK}
E_n(K; \lambda,\mu):= \inf\limits_{\phi_1,\ldots,\phi_{n-1}\in
L^{2,o}(\T^2)}\,\,\sup\limits_{{\psi}
\in[\phi_1,\ldots,\phi_{n-1}]^\perp,\,\|{\psi}\|=1}
({H}_{\lambda\mu}(K){\psi},{\psi}).
\end{equation}
By the minimax principle, $e_n(K;\lambda,\mu)\le \cE_{\min}(K)$ and
$E_n(K;\lambda,\mu)\ge \cE_{\max}(K).$ Since, the rank of
$V_{\lambda\mu}$ does not exceed six, by choosing suitable elements
$\phi_1$, $\phi_2$, $\phi_3$, $\phi_4$, $\phi_5$, and $\phi_6$ from
the range of $V_{\lambda\mu}$  one concludes that
$e_n(K;\lambda,\mu) = \cE_{\min}(K)$ and $E_n(K;\lambda,\mu) =
\cE_{\max}(K)$ for all $n\ge7.$

\begin{lemma}\label{lem:monoton_xos_qiymat}
Let $n\ge1$ and $i\in\{1,2\}.$  For every fixed $K_j\in\T,$
$j\in\{1,2\}\setminus\{i\},$  the map
$$
K_i\in\T \mapsto \cE_{\min}((K_1,K_2)) - e_n((K_1,K_2);\lambda,\mu)
$$
is non-increasing in $(-\pi,0]$ and non-decreasing in $[0,\pi]$.
Similarly, for every fixed $K_j\in\T,$ $j\in\{1,2\}\setminus\{i\},$
the map
$$
K_i\in\T \mapsto E_n((K_1,K_2);\lambda,\mu) - \cE_{\max}((K_1,K_2))
$$
is non-increasing in $(-\pi,0]$ and non-decreasing in $[0,\pi]$.
\end{lemma}

\begin{proof}
Without loss of generality we assume that $i=1.$ Given ${\psi}\in
L^{2,o}(\T^2)$ consider
$$
(({H}_0(K) - \cE_{\min}(K)){\psi},{\psi})=\int_{\T^2}
\sum\limits_{i=1}^2 \cos\tfrac{K_i}{2}\,\big(1-\cos
q_i\big)|\psi(q)|^2\,\mathrm{d} q, \quad K:=(K_1,K_2).
$$
Clearly, the map $K_1\in\T\mapsto (({H}_0(K) -
\cE_{\min}(K)){\psi},{\psi})$ is non-decreasing in $(-\pi,0]$ and is
non-increasing in $[0,\pi].$ Since ${V}_{\lambda\mu}$ is independent
of $K,$ by definition of $e_n(K;\lambda,\mu)$ the map
$K_1\in\T\mapsto e_n(K;\lambda,\mu) - \cE_{\min}(K)$ is
non-decreasing in $(-\pi,0]$ and is non-increasing  in $[0,\pi].$

The case of $K_i\mapsto E_n(K;\lambda,\mu) - \cE_{\max}(K)$ is
similar.
\end{proof}

\noindent{\textit{Proof of Theorem \ref{teo:discr_Kne0}.}} By Lemma
\ref{lem:monoton_xos_qiymat} for any $K\in\T^2$ and $m\ge1$ we have
\begin{equation}\label{eq:min_eigen0}
0\le \cE_{\min}(0) - e_m(0;\lambda,\mu) \le \cE_{\min}(K) -
e_m(K;\lambda,\mu)
\end{equation}
and
\begin{equation*}
E_m(K;\lambda,\mu) - \cE_{\max}(K) \ge E_m(0;\lambda,\mu) -
\cE_{\max}(0) \ge 0.
\end{equation*}
By the assumption, $ e_n(0;\lambda,\mu)$ is a discrete eigenvalue of
${H}_{\lambda\mu}(0)$ for some $\lambda,\mu\in\R.$ Thus,
$\cE_{\min}(0) - e_n(0;\lambda,\mu)>0,$ and hence, by
\eqref{eq:min_eigen0} and \eqref{eq:essential_spectrum}
$e_n(K;\lambda,\mu)$ is a discrete eigenvalue of
${H}_{\lambda\mu}(K)$ for any $K\in\T^2.$ Since
$e_1(K;\lambda,\mu)\le \ldots \le e_n(K;\lambda,\mu)<\cE_{\min}(K),$
it follows that ${H}_{\lambda\mu}(K)$ has at least $n$ eigenvalue
below its essential spectrum. The case of $E_n(K;\lambda,\mu)$ is
similar.\,\,\hfill $\qed$
\smallskip

\noindent{\textit{Proof of Theorem \ref{teo:disc_quasiK}}} can be
obtained by combining Theorem \ref{teo:constant} with Theorem
\ref{teo:disc_quasi0}.\,\hfill $\square$

\appendix
\section{}
\label{sec:append_A}

Let us first prove the corresponding asymptotical formula in
\eqref{asimp_function0} for $a(z)$ as $z\nearrow0$. We start with an
elementary observation that
\begin{equation}
\label{aasym} a(z) =\frac{1}{8\pi^2}\int\limits_
{\mathbb{T}^2}\frac{(\sin p_{1}+\sin p_2)^2\
\mathrm{d}p}{\cE_{0}(p)-z}=\frac{1}{4\pi^2}\int\limits_
{\mathbb{T}^2}\frac{\sin^2 p_{1}\ \mathrm{d}p}{\cE_{0}(p)-z},
\end{equation}
taking into account that the function
$\cE_{0}(p)=2\epsilon\bigl((p_1,p_2)\bigr)=2\sum\limits_{i=1}^2
\big(1-\cos p_i)$ is invariant under permutations of the components
$p_1$ and $p_2$. One also notices that, for $|A|>1$,
\begin{equation}\label{1sinus}
\int_{\T}\frac{\sin^2 t \mathrm{d} t}{A - \cos t} =
\begin{cases}
\pi\sqrt{A^2-1}-{\pi}A & \text{if $A<-1,$}\\
{\pi}A-\pi\sqrt{A^2-1} & \text{if $A>1$.}
\end{cases}
\end{equation}
Since $2-\frac{z}{2}-\cos q>1$ for  any  $q\in \T$ and $z<0$, by
using $\int_\T \cos q\,\,\mathrm{d} q=0$ and \eqref{1sinus}, one
then finds from \eqref{aasym} that
\begin{align}
\label{Ia} &\lim\limits_{z\nearrow
0}a(z)=1-\frac{1}{4\pi^2}\,\,I_{a},
\end{align}
where
$$
I_a={\pi}\int
\limits_{\mathbb{T}}\sqrt{(2-\cos\,q)^2-1}\,\,\mathrm{d}q.
$$
Performing in $I_a$ the change of variables $v=\sqrt{2\left(\tan
\frac{q}{2}\right)^2+1}$, one obtains
\begin{align}
\label{Ia1} I_a=&
{16\pi}\int\limits_{1}^{+\infty}\frac{v^2}{(v^2+1)^2}dv =
{16\pi}\left(\frac{\arctan
v}{2}-\frac{v}{2(v^2+1)}\right)\bigg|_1^{+\infty}=2\pi^2+4\pi.
\end{align}
Combining \eqref{Ia} and \eqref{Ia1} yields the required asymptotics
\begin{align}
\label{Ia3} &\lim\limits_{z\nearrow 0}a(z)=\frac{\pi-2}{2\pi},
\end{align}
thus, completing its proof.

The remaining asymptotical formulae in
\eqref{asimp_function0}--\eqref{asimp_functions4}  are derived
analogously and, thus, we skip the respective computation. We
remark, however, that it is somewhat reduced due to Remark
\ref{Rabcd}. \hfill $\square$ \vspace{5mm}

{\bf Acknowledgments}. The authors thank the anonymous referees for
important remarks and suggestions. The authors acknowledge support
of this research by Ministry of Innovative Development of the
Republic of Uzbekistan (Grant No. FZ–20200929224).


\end{document}